\RequirePackage{fix-cm}
\documentclass[smallextended]{svjour3}       
\smartqed  

\usepackage{times}
\usepackage{amssymb}
\usepackage{url}
\usepackage{tabularx}
\usepackage{multirow}
\usepackage{calc}
\usepackage[table]{xcolor}
\usepackage{enumitem}
\usepackage{wrapfig}
\usepackage{verbatim}
\usepackage{tikz,color}
\usepackage{booktabs}
\usepackage{hyperref}
\hypersetup{
     colorlinks   = true,
     citecolor    = blue
}
\usepackage{xspace}
\usepackage{amsmath}
\usepackage{hyphenat}
\usepackage[right]{lineno}

\newcommand{\np}{NP}
\newcommand{\poly}{P}
\newcommand{\nph}{{\np}-hard}
\newcommand{\nphns}{{\np}-hardness}
\newcommand{\nphshort}{{\np}-h}
\newcommand{\threesat}{3-SAT}

\newcommand{\mathxxxcssize}{\kappa}  
\newcommand{\mathxxxcins}{\mathcal{I}}   
\newcommand{\mathgcaiins}{\mathcal{E}_{\mathxxxcins}}  
\newcommand{\mathGCPIthreetwoins}{\mathcal{E}_{\mathcal{I}}}

\newcommand{\fpt}{{FPT}}
\newcommand{\no}{NO}
\newcommand{\yes}{YES}
\newcommand{\noins}{{\no}-instance}
\newcommand{\yesins}{{\yes}-instance}

\newcommand{\conr}{consent rule}
\newcommand{\edge}[2]{\{#1,#2\}}
\newcommand{\including}{{adding}}
\newcommand{\excluding}{{deleting}}
\newcommand{\gcii}{{GCAI}}
\newcommand{\gcei}{{GCDI}}
\newcommand{\gcpi}{GCPI}

\newcommand{\unsurehidden}[1]{{#1}} 
\renewenvironment{proof}[1][{\it{Proof}}]{\noindent\textbf{#1} }{\hfill \qed\medskip}

\let\citeasnoun=\cite

\begin{document}

\title{How hard is it to control a group?\thanks{A preliminary version of this paper
was presented at the 6th International Workshop on Computational
Social Choice (COMSOC 2016).}
}

\subtitle{}


\author{Yongjie Yang         \and
        Dinko Dimitrov
}


\institute{Yongjie Yang\at
              Chair of Economic Theory, Saarland University, Germany\\
              \email{yyongjiecs@gmail.com}           
           \and
         Dinko Dimitrov \at
Chair of Economic Theory, Saarland University, Germany\\
\email{dinko.dimitrov@mx.uni-saarland.de}
}

\date{\today}

\maketitle

\begin{abstract}
We consider group identification models in which the aggregation of
individual opinions concerning who is qualified in a given society
determines the set of socially qualified persons. In
this setting, we study the extent to which social qualification can be
changed when societies expand, shrink, or partition themselves. The answers
we provide are with respect to the computational complexity of the
corresponding control problems and fully cover the class of consent
aggregation rules introduced by Samet \& Schmeidler (2003) as well as
procedural rules for group identification. We obtain both polynomial-time solvability results and {\nphns} results.
In addition, we also study these problems from the parameterized complexity
perspective, and obtain some fixed-parameter tractability results.

\end{abstract}

\section{Introduction}
Group decision making plays an important role in multi-agent systems.
Imagine for instance a set $N$ of agents who have to determine those among them who are
eligible or qualified to complete a task. In such a case the view of all
agents has to be taken into account and a rule for the selection of a subset
of $N$ should be specified. In this paper we consider a specific decision
making model in which each individual qualifies or disqualifies every
individual in $N$, and then a social rule is applied to select the socially
qualified individuals.  This model has been
widely studied under the name \textit{group identification} in
economics~(see~\cite{DBLP:series/sfsc/Dimitrov11} for a survey). In particular, the {liberal rule}, the class of {consent rules}, the
consensus-start-respecting rule (CSR), and the liberal-start-respecting rule
(LSR) have been axiomatically characterized in the literature~\cite{DBLP:journals/mss/DimitrovSX07,DBLP:journals/mss/Nicolas07,KasherRwhoisj,DBLP:journals/geb/Miller08,DBLP:journals/jet/SametS03}. Under the liberal rule, an individual is socially
qualified if and only if this individual qualifies {{her}self}. Each rule in
the class of consent social rules is characterized by two positive integers $s$ and $t$. Specifically, if an individual qualifies {{her}self}, then this
individual is socially qualified if and only if there are at least $s-1$
other individuals who also qualify {her}. On the other hand, if the
individual disqualifies {{her}self}, then this individual is not socially
qualified if and only if there are at least $t-1$ other individuals who also
disqualify {her}. Finally, the CSR and the LSR social rules recursively
determine the socially qualified individuals. In the beginning, the set $K^{L}$ of individuals each of whom qualifies {{her}self} are considered LSR
socially qualified, while the set $K^{C}$ of individuals each of whom is
qualified by all individuals are considered CSR socially qualified. Then, in
each iteration for the social rule LSR (resp.\ CSR), an individual $a$ is
added to $K^{L}$ (resp.\ $K^{C}$) if there is an individual in $K^{L}$ (resp.\ $K^{C}$) qualifying $a$. The iteration terminates when no new individual
can be added to $K^{L}$ (resp.\ $K^{C}$), and the socially qualified
individuals are the ones in $K^{L}$ (resp.\ $K^{C}$).

In this paper, we consider the problems where an external (strategic) agent
has an incentive to control the results by either adding some individuals ({\gcii}), or deleting some individuals ({\gcei}),
or partitioning the set of individuals ({\gcpi}). In particular, in each problem
the external agent has a subset $S$ of individuals and the goal of the external agent is to make all individuals in $S$
 socially qualified  (see Section~\ref{sec-problem-formulations} for the precise definitions of the {\gcii, \gcei}, and {\gcpi} problems). We study the complexity of these problems for the {liberal rule},
the class of {consent rules}, and the CSR and LSR rules, aiming to show how hard a given problem is. We achieve both
polynomial-time solvability results and {\nphns} results for these
problems. In particular, we obtain dichotomy results for all problems
considered in this paper for consent rules, with respect to the values of $s$
and $t$. In addition, we study the {\nph} problems from the
parameterized complexity point of view, and obtain several fixed-parameter
tractability results, with respect to $|S|$.
See Table~\ref{tab:dichtomyforconsetrules} for a summary of our main findings.

\begin{table*}[h!]
  \caption{A summary of the complexity of the {\gcii}, {\gcei}, and {\gcpi} problems. In the table, ``{\nphshort}'' stands for ``{\nph}'', ``P'' stands for ``polynomial-time solvable'', and ``I'' stands for  ``immune''.
The {\nphns} results with the symbol ``+'' next to them mean that the problems are fixed-parameter tractable ({\fpt}) with respect to $|S|$, where $S$ is the set of individuals the strategic agent wants to make socially qualified (for the precise definition of $S$, see the definitions of {\gcii}, {\gcei}, and {\gcpi} in Section~\ref{sec-problem-formulations}).
}
  \begin{center}
    \begin{tabular}{lllllllllll}\toprule
      &    \multicolumn{6}{l}{{{\conr}s} $f^{(s,t)}$}  &   &&   \\ \cmidrule{2-8}

      &                               \multicolumn{3}{l}{$s=1$}  && \multicolumn{3}{l}{$s\geq 2$} & &$f^{CSR}$& $f^{LSR}$\\ \cmidrule{2-4}\cmidrule{2-4} \cmidrule{6-8}

      & $t=1$ &   $t=2$   &  $t\geq 3$   &&  $t=1$  & $t=2$ & $t\geq 3$ && & \\ \cmidrule{2-4} \cmidrule{6-8} \cmidrule{10-11}

      {\gcii}       & I&  I &  I  &  &  {{\nphshort} {\scriptsize{(+)}}} & {\nphshort} {\scriptsize{(+)}} & {\nphshort} {\scriptsize{(+)}} &&{\nphshort}& {\nphshort} \\

      {\gcei}       & I&  P & {\nphshort} {\scriptsize{(+)}} && I & P & {\nphshort} {\scriptsize{(+)}}&  &P & I \\

      {\gcpi} & I& {\nphshort} & {\nphshort}& &I& {\nphshort}  &{\nphshort} && ? &  I\\ \bottomrule

    \end{tabular}
  \end{center}

  \label{tab:dichtomyforconsetrules}
\end{table*}

To the best of our knowledge, group identification as a classic model for
identifying socially qualified individuals has not been studied from the
complexity point of view~\footnote{After the workshop version of the current paper, Erd{\'{e}}lyi, Reger, and Yang studied the complexity of destructive control, bribery, and possibly/necessarily socially qualified individuals problems in group identification~\cite{AAMAS17ErdelyiRYBriberyControlGroupIdentification,DBLP:conf/aldt/ErdelyiRY17}.}. The words ``control by
adding/deleting/partitioning of'' in the names of the group identification control problems are
reminiscent of many strategic voting problems, such as control by
adding/deleting/partitioning of voters/candidates, which have been extensively
studied in the
literature~\cite{Bartholdi92howhard,%
DBLP:journals/jair/FaliszewskiHHR09,%
handbookofcomsoc2016Cha7FR,%
DBLP:journals/jair/MeirPRZ08,%
Yang2014,%
AAMAS15YangGmultipeak}.
In a voting system, we have a set of candidates
and a set of voters. Each voter casts a vote, and a voting correspondence is
used to select a subset of candidates. From this standpoint, group
identification can be considered as a voting system where the individuals
are both voters and candidates. Nevertheless, group identification differs
from voting systems in many significant aspects. First, the goal of a voting
system is to select a subset of candidates, who are often called winners
since they are considered as more competitive or outstanding compared with
the remaining candidates for some specific purpose. Despite that the goal of
group identification is also to identify a set of individuals (socially
qualified individuals) from the entire set of individuals, it does not imply
that socially qualified individuals are more competitive or outstanding than
the remaining individuals. For instance, in situations where we want to
identify left-wing party members among a group of people, the model of group
identification is more suitable. In other words, group identification is
closer to a classification model. Second, as voting systems aim to select a
subset of competitive candidates for some special purpose, more often than
not, the number of winners is pre-decided (e.g., in a single-winner voting,
exactly one candidate is selected as the winner). As a consequence, many
voting systems need to adopt a certain tie-breaking method to break the tie
when many candidates are considered equally competitive. However, group
identification does not need a tie-breaking method, since there is no
bound on the number of socially qualified individuals.

It is also worth pointing out that the classic voting system Approval, which has been widely studied in the
literature~\cite{baumeisterapproval09,%
Fishburm81,%
DBLP:journals/ai/HemaspaandraHR07,%
handbookapprovalvoting,%
DBLP:conf/icaart/Lin11,%
DBLP:journals/jcss/YangG17}, has the flavor of group identification.
In Approval voting, each voter
approves or disapproves each candidate. Thus, each voter's vote is
represented by a 1-0 vector, where the entries with 1s (resp.\ 0s) mean that
the voter approves (resp.\ disapproves) the corresponding candidate. The
winners are among the candidates which get the most approvals. If the voters
and candidates are the same group of individuals, then it seems that
Approval voting is a social rule. Nevertheless, as discussed above, Approval
voting is more often considered as a single-winner voting system and thus
needs to utilize a tie-breaking method. Recently, several variants of
Approval voting have been studied as multi-winner voting systems. However,
the number of winners is bounded by (or exactly equal to) an integer~(see, e.g.,~\cite{DBLP:conf/atal/AzizGGMMW15,Kilgour2014Marshall}). Moreover, to the best of our knowledge, complexity of control
by adding/deleting/partitioning of voters/candidates has not been studied for
Approval voting when the voters and candidates coincide.

Recently, multiwinner voting where  the number of winners is not fixed has also been studied (see, e.g.,~\cite{DBLP:journals/scw/BergaBMN04,multiwinnerwithvariablewinnersFST2017,Kilgour2016,YangWangIJCAI2018}). However, these rules are completely different from what we study in the paper. Moreover, to the best of our knowledge, to date only the very recent papers~\cite{multiwinnerwithvariablewinnersFST2017} and~\cite{YangWangIJCAI2018} (which appeared after the workshop version of our paper) considered such multiwinner voting from the complexity point of view. However, they mainly considered the winner determination problem while we consider control problems.

The rest of the paper is organized as follows. In Section~\ref{sec_notation} we formally introduce the
studied social rules, the {\nph} problems we make use of in our proofs, as
well as the studied group control problems. Section~\ref{sec_consent_complexity} is then devoted to the
study of the control problems when the corresponding aggregation rule is a
consent rule, while in Section~\ref{sec_procedure_complexity} we study these problems with respect to the
procedural rules. We assume then in Section~\ref{sec_fpt} that the size of the group of
individuals to be made socially qualified is bounded and study the
fixed-parameter tractability of the group control problems. We conclude our work and offer some directions for future research in Section~\ref{sec_conclusion}.

\section{Basic notation and definitions}
\label{sec_notation}
Throughout this paper we will need the
following basic notions and concepts.

\subsection{Social rules}
Let $N$ be a set of individuals. We assume that each
individual $a\in N$ has an opinion about who from the set $N$ possesses a certain
qualification and who does not. For $a^{\prime }\in N$, we write $\varphi (a,a^{\prime })=1$ to denote the fact that $a$ qualifies $a^{\prime }$, and $\varphi (a,a^{\prime })=0$ to denote the fact that $a$ disqualifies $a^{\prime }$. The mapping $\varphi
:N\times N\rightarrow \left\{ 0,1\right\} $ is called a \textit{profile}
over $N$. A \textit{social rule} is a function $f$ assigning a subset~$%
f(\varphi ,T)\subseteq T$ to each pair $(\varphi ,T)$ consisting of a
profile $\varphi $ over $N$ and a subset $T\subseteq N$. We call the
individuals in $f(\varphi ,T)$ the \textit{socially qualified individuals}
of $T$ with respect to $f$ and $\varphi $.\medskip\

In what follows we focus our analysis on the class of consent rules
introduced by Samet and Schmeidler~\citeasnoun{DBLP:journals/jet/SametS03} and on the procedural rules for
group identification axiomatically studied in~\citeasnoun{DBLP:journals/mss/DimitrovSX07}.
\bigskip

{\bf{Consent rules $f^{(s,t)}$}}. Each {\conr} $f^{(s,t)}$ is specified by two positive integers $s$ and $t$ such that
 for every $T\subseteq N$ and every individual $a \in T$,
\begin{enumerate}
\item  if $\varphi(a,a)=1$, then $a \in f^{(s,t)}(\varphi, T)$ if
      and only if $|\{a' \in T \mid \varphi(a', a) =1\}|\geq s$, and
\item if $\varphi(a, a)=0$, then $a \not\in f^{(s,t)}(\varphi, T)$
      if and only if $|\{a' \in T \mid \varphi(a', a) =0\}|\geq t$.
\end{enumerate}

 Notice that if an individual qualifies (disqualifies) herself but lacks the qualification (disqualification) of at least $s-1$ ($t-1$) other individuals, then she will be socially disqualified/qualified. The two positive integers $s$ and $t$ are referred to as the {\it{consent quotas}} of the
  rule $f^{(s,t)}$.
  It is worth mentioning that in the original definition of {\conr}s by Samet and Schmeidler~\citeasnoun{DBLP:journals/jet/SametS03} there is an additional
  condition $s+t\leq n+2$ for consent quotas $s$ and $t$ to satisfy, where $n$ is the number of individuals. Indeed, the condition $s+t\leq n+2$ is crucial for the {\conr}s
  to satisfy the \textit{monotonicity property} requiring a socially qualified individual $a$
  to be still socially qualified when someone who disqualifies $a$ turns
  to qualify $a$. Since our paper is mainly concerned with the
  computational complexity of group control problems, we drop this condition from the definition
  of the {\conr}s (we indeed achieve results for a more general class of social rules that
  encapsulates the original {\conr}s defined in the work of Samet and Schmeidler~\citeasnoun{DBLP:journals/jet/SametS03}). When studying the group control problems for the consent
rules $f^{(s,t)}$ we assume that the consent quotas $s$ and $t$ remain the
same, that is, they do not change after {\including} new individuals, {\excluding}
old ones, or partitioning the set of individuals. Finally, we would like to point out that the consent rule $f^{(1,1)}$ is also referred to as the liberal rule in the literature~\cite{DBLP:journals/jet/SametS03}.

\bigskip

{\bf{Consensus-start-respecting rule $f^{CSR}$}}. For every $T\subseteq N$, this
  rule determines the socially qualified individuals
  iteratively. First, all individuals who are qualified by everyone in the society are considered socially qualified. Then, in each
  iteration, all individuals who are qualified by at least one of
  the currently socially qualified individuals are added to the set of
  socially qualified individuals. The iterations terminate when no new
  individual is added. Formally, for every $T\subseteq N$, let
  \begin{displaymath}
    K_0^C(\varphi, T)=\{a \in T\mid \forall a' \in T,\ \varphi(a', a) =1\}.
  \end{displaymath}
  For each positive integer $\ell=1,2,\dots$, let
  \begin{displaymath}
    K_\ell^C(\varphi, T)=K_{\ell-1}^C(\varphi, T) \cup \{a \in T \mid
    \exists a' \in K_{\ell-1}^C(\varphi, T),\ \varphi(a', a) =1\}.
  \end{displaymath}
  Then $f^{CSR}(\varphi, T)=K_\ell^C(\varphi, T)$ for some $\ell$ such
  that $K_\ell^C(\varphi,T)=K_{\ell-1}^C(\varphi, T)$.

\bigskip
{\bf{Liberal-start-respecting rule $f^{LSR}$}}. This rule
  is similar to $f^{CSR}$ with the only difference that the initial
  socially qualified individuals are those who qualify themselves. In
  particular, for every $T\subseteq N$, let
  \begin{displaymath}
    K_0^L(\varphi, T)=\{a \in T\mid \varphi(a, a)=1\}.
  \end{displaymath}
  For each positive integer $\ell=1,2,\dots$, let
  \begin{displaymath}
    K_\ell^L(\varphi, T)=K_{\ell-1}^L(\varphi, T) \cup \{a \in T \mid
    \exists a' \in K_{\ell-1}^L(\varphi, T),\ \varphi(a', a) =1\}.
  \end{displaymath}
  Then $f^{LSR}(\varphi, T)=K_\ell^L(\varphi, T)$ for some $\ell$ such
  that $K_\ell^L(\varphi, T)=K_{\ell-1}^L(\varphi, T)$.
\bigskip

Clearly, when $K_0^C$ (resp.\ $K_0^L$) is empty we have that $f^{CSR}(\varphi, T)=\emptyset$ (resp.\ $f^{LSR}(\varphi, T)=\emptyset$).

\subsection{Group control}
\label{sec-problem-formulations}
Let us now formally state
the three group control problems we study. In the following, let $f$ be a social rule.

\begin{description}\itemsep=5pt
\item[Group Control by Adding Individuals] ({\gcii})
  \begin{itemize}
  \item[\textit{Input}:] A 5-tuple $(N, \varphi, S, T, k)$
    of  a set $N$ of individuals, a profile~$\varphi$ over
    $N$, two nonempty subsets $S, T \subseteq N$ such
    that~$S \subseteq T$ and $S \not\subseteq f(\varphi, T)$, and a
    positive integer $k$.
  \item[\textit{Question}:] Is there a subset
    $U\subseteq N \setminus T$ such that $|U|\leq k$ and
    $S\subseteq f(\varphi, T \cup U)$?
  \end{itemize}

\item[Group Control by Deleting Individuals] ({\gcei})
  \begin{itemize}
  \item[\textit{Input}:] A 4-tuple $(N, \varphi, S, k)$ of a
     set $N$ of individuals, a profile~$\varphi$
    over $N$, a nonempty subset $S\subseteq N$ such that
    $S \not\subseteq f(\varphi, N)$, and a positive integer~$k$.
  \item[\textit{Question}:] Is there a subset
    $U \subseteq N \setminus S$ such that $|U|\leq k$ and
    $S\subseteq f(\varphi, N \setminus U)$?
  \end{itemize}

\item[Group Control by Partitioning of Individuals] ({\gcpi})
  \begin{itemize}
  \item[\textit{Input}:]
    A 3-tuple $(N, \varphi, S)$ of a set $N$
    of individuals, a profile $\varphi$ over $N$, and a nonempty subset
    $S\subseteq N$ such that $S\not\subseteq f(\varphi, N)$.
  \item[\textit{Question}:] Is there a subset $U \subseteq N$ such that
    $S\subseteq f(\varphi, V)$ where
    \[V =f(\varphi, U)\cup f(\varphi, N \setminus U)?\]
  \end{itemize}
\end{description}

We say that a social rule is \textit{immune} to a control type if it is
impossible to make a socially disqualified individual $a \in S$
socially qualified by carrying out the corresponding operations (i.e., {\including} individuals, {\excluding} individuals, or partitioning the set
of individuals).
If a social rule is not immune to a control type involved in a problem defined above, we say it is {\it{susceptible}} to the control type.

\subsection{Some {\nph} problems}

Our {\nphns} results in this paper are shown by efficient reductions from the following {\nph}
problems: a restricted version of {\sc{Exact Cover by Three-Sets}} (RX3C), {\sc{Labeled Red-Blue
Dominating Set (LRBDS)}}, and {\sc{3-Satisfiability ({\threesat})}}. The formal definitions of these problems are as follows.


\begin{description}
\item [{\sc{Restricted version of Exact Cover by Three-Sets}}] (RX3C)
  \begin{itemize}
  \item[\textit{Input}:] A finite set $X$ with $|X|=3{\mathxxxcssize}$
    for some positive integer $\mathxxxcssize$ and a collection
    $\mathcal C$ of 3-subsets of $X$ such that every $x\in X$ occurs in exactly three 3-subsets in $\mathcal C$. So, it holds that $|\mathcal{C}|=3{\mathxxxcssize}$.
  \item[\textit{Question}:] Is there a subcollection
    $\mathcal{C}'\subseteq \mathcal C$ such that
    $|\mathcal{C}'|={\mathxxxcssize}$ and each $x\in X$ appears in
    exactly one set of~$\mathcal{C}'$?
  \end{itemize}
\end{description}

The {\nphns} of {RX3C} is shown in~\cite{DBLP:journals/tcs/Gonzalez85} (Theorem A.1).

\medskip

In order to
state the second {\nph} problem, we will need the following basic notions
from graph theory. We consider only undirected graphs. A {\it{graph}} is a tuple $(W,E)$ where $W$ is the
{\it{vertex set}} and $E$ is the {\it{edge set}}.
A vertex $v$ {\it{dominates}} a vertex $u$ if there is an edge between $v$ and $u$.
A vertex subset $A$ dominates another vertex subset $B$, if for every
vertex $u\in B$ there is some vertex $v\in A$ that dominates $u$.
An {\it{independent set}} $I$ of a graph is a vertex subset such that there is no edge between each pair of vertices in $I$.
A {\it{bipartite graph}} is a graph
whose vertex set can be partitioned into two independent sets. We denote by $(L\uplus R, E)$ a bipartite graph with $(L,R)$ being a partition of its vertex set such that both $L$ and $R$ are independent sets. We refer to the textbook of West~\cite{Douglas2000} for further details on graphs.

We can now state the second {\nph} problem we will use in the next sections.

\begin{description}
\item[{\sc{Labeled Red-Blue Dominating Set}}] (LRBDS)
  \begin{itemize}
  \item[\textit{Input}:] A bipartite graph $G=(R\uplus B,E)$, where each vertex in $R$ has a label from $\{1,2,\dots,k\}$. For every $i\in \{1,2,\dots,k\}$, let $R_i$ be the set of all vertices in $R$ that have label $i$.
  \item[\textit{Question}:] Is there a subset $W\subseteq R$ such that $|W\cap R_i|\leq 1$ for every $i\in \{1,2,\dots,k\}$ and $W$ dominates $B$?
  \end{itemize}
\end{description}

The following additional notions will be needed as to state the
third {\nph} problem mentioned above. A {\it{Boolean variable}} $x$ takes either the value 1 or 0. Let $X$ be a set of Boolean variables. If $x\in X$, then $x$ and $\bar{x}$ are {\it{literals}} over $X$. A {\it{clause}} $c$ over $X$ is a set of literals over $X$. A {\it{truth assignment}} is a function $\varrho: X\rightarrow \{0,1\}$. A clause $c$ is {\it{satisfied}} under a truth assignment $\varrho$ if and only if there is an $x$ in $c$ such that $\varrho(x)=1$, or a $\bar{x}$ in $c$ such that $\varrho(x)=0$. The {\sc{3-Satisfiability}} problem defined below is a famous {\nph} problem~\cite{DBLP:conf/stoc/Cook71,garey}.

\begin{description}
\item[{\sc{3-Satisfiability}}] (\threesat)
  \begin{itemize}
  \item[\textit{Input}:] A set $X$ of Boolean variables, and a collection $C$ of clauses over $X$ such that each clause includes exactly three literals.
  \item[\textit{Question}:] Is there a truth assignment $\varrho: X\rightarrow \{0,1\}$ under which all clauses in $C$ are satisfied?
  \end{itemize}
\end{description}

\section{Consent rules}
\label{sec_consent_complexity}
We start our analysis by
investigating the group control problems with respect to consent rules.
Section~\ref{subsec_immune_consent_rules} describes the subclass of rules that turn out to be immune to
(some of) these control types. In Sections~\ref{subsec_polynomial_solvable_rules} and~\ref{subsec_NP-hard_rules} we then explore the
computational complexity of group control problems for consent rules that
are susceptible to the corresponding control types.

\subsection{Immune consent rules}
\label{subsec_immune_consent_rules}
The intrinsic property of the consent rule $f^{(1,1)}$ is that it completely
leaves to each individual to determine her own social qualification.
Put it another way, whether an individual is socially qualified
is independent of the opinions of any other individual.
As a consequence, the answers to the question whether an individual is
socially qualified before and after the operations in the corresponding group control problems
are the same, as implied by the following theorem.

\begin{theorem}
  \label{thm:liberalimmune}
  The consent rule $f^{(1,1)}$ is immune to {\gcii}, {\gcei}, and {\gcpi}.
\end{theorem}

\begin{proof}
  Consider instances of the {\gcii}, {\gcei}, and {\gcpi} problems with the consent rule $f^{(1,1)}$ as their
  social rule.  Notice that, as an assumption on instances, $S\nsubseteq
f^{(1,1)}(\varphi ,T)$ with $S\subseteq T$ is imposed in {\gcii}, while $%
S\nsubseteq f^{(1,1)}(\varphi ,N)$ is imposed in both {\gcei} and {\gcpi}. Due to the definition of $f^{(1,1)}$, each of the above
  assumptions implies that there exists an individual $a \in S$ such
  that $\varphi(a,a) = 0$, and hence, $a \not\in f^{(1,1)}(\varphi, V)$ for
  every $V \subseteq N$. It follows that
  $S \not\subseteq f^{(1,1)}(\varphi, V)$ for every $S\subseteq V \subseteq N$.
  Therefore, for all instances with $f^{(1,1)}$ as the social rule, the
  answers to {\gcii}, {\gcei}, and {\gcpi} are always ``{\no}''. This
  completes the proof.
\end{proof}

Let us now turn
to consent rules $f^{(s,1)}$ and $f^{(1,t)}$ with $s,t\geq 2$. In order to
change the social status of an individual $a$ from disqualified to qualified
when $f^{(s,1)}$ is applied, one needs the number of supporters of $a$ to
increase, but neither {\excluding} individuals nor partitioning the set of
individuals seems to create additional support. On the other hand, when
the applied rule is $f^{(1,t)}$, the number of people who disqualify $a$
does not decrease when adding individuals. Our next result confirms this
intuition with respect to the corresponding control problems.

\begin{theorem}
  \label{thm:consentimmune}
  Every {\conr} $f^{(s,1)}$ with $s\geq 2$ is immune to {\gcei} and {\gcpi}, and every
  {\conr} $f^{(1,t)}$ with $t\geq 2$ is immune to {\gcii}.
\end{theorem}

\begin{proof}
We first consider {\conr}s $f^{(s,1)}$ with $s\geq 2$.
  Let $a\in S$ be an individual who is not socially qualified, i.e.,
  $a \not\in f^{(s,1)}(\varphi, N)$. We distinguish between two cases.
  \begin{description}
  \item[Case $\varphi(a,a) =1$:] There are at most $s-1$ individuals
    in $N$ qualifying individual $a$, i.e.,
    $|\{a' \in N \mid \varphi(a', a) =1\}| < s$, and
    thus~$|\{a' \in V \mid \varphi(a', a) =1\}| < s$ for every
    $V \subseteq N$. Therefore, it is impossible to make individual
    $a \in S$ socially qualified by {\excluding} or partitioning the set of individuals.

  \item[Case $\varphi(a, a) = 0$:] By definition, each consent rule $f$
    satisfies $f(\varphi, V) \subseteq V$, and hence,
    $S \not\subseteq f^{(s,1)}(\varphi, V)$ if $S \not\subseteq V$;
    otherwise, when~$S \subseteq V$, we have
    $|\{a' \in V \mid \varphi(a', a) =0\}| \ge 1$ from
    $\varphi(a, a) = 0$, which implies that
    $a \not\in f^{(s,1)}(\varphi, V)$.  Hence, it is impossible to make
    individual $a$ socially qualified by {\excluding} or partitioning the set of individuals,
    i.e., $S \not\subseteq f^{(s,1)}(\varphi, V)$ for any
    $V \subseteq N$.
  \end{description}
  Therefore, for each instance with $f^{(s,1)}$ as
  its social rule, the answers to {\gcei} and {\gcpi} are always
  ``{\no}''.

Consider now the {\conr} $f^{(1,t)}$.
  Let $a\in S$ be an individual which is not socially qualified, that
  is $a\not\in f^{(1,t)}(\varphi, T)$.  This implies that
  $\varphi(a,a)=0$ and, moreover, there are at least $t$
  individuals $a'$ (including $a$) in $T$ such that $\varphi(a',a)=0$.
  Therefore, no matter which individuals the set $U$ includes, there
  will be still at least $t$ individuals $a' \in T\cup U$ such that
  $\varphi(a', a)=0$, implying that $a$ is still not socially
  qualified.
\end{proof}

In the remaining
subsections we show how the interplay between the consent quotas $s$ and $t$
shapes the extent to which the corresponding consent rules are susceptible to
all the three group control types.

\subsection{Polynomial-time solvability}
\label{subsec_polynomial_solvable_rules}
We start with the {\gcei} problem for consent rules $f^{(s,2)}$ with $s\geq 1$. In order to show that
every such rule is not immune to this group control type, we need only to
give an instance where one can make all individuals in $S$ socially
qualified by {\excluding} a limited number of individuals, given that not all
individuals in $S$ are socially qualified in advance. To this end, consider
an instance $\left(N=\left\{ a,b\right\} ,\varphi ,S=\left\{
a\right\} ,k=1\right) $ where $\varphi (a,a)=\varphi (b,a)=0$. It is clear
that one can make $a$ socially qualified by {\excluding} $b$ from the instance.
As our next result reveals, it is practically tractable for a designer to
control a group identification procedure by {\excluding} individuals, provided
that the social rule is $f^{(s,2)}$.

\begin{theorem}
\label{thm:GCEIconsentstwopolynomial}
 The {\gcei} problem for every {\conr} $f^{(s,2)}$ with $s\geq 1$ is polynomial-time solvable.
\end{theorem}

\begin{proof}
  Let $L = \{a \in S \mid \varphi(a, a) = 1\}$ and
  $\bar L = S \setminus L = \{a \in S \mid \varphi(a, a) = 0\}$.  For
  each~$a\in \bar L$, let~$U_a \subseteq N$ be the set of individuals
  each of whom is outside $S$ and disqualifies $a$, i.e.,
  $U_a= \{ a' \in N\setminus S \mid \varphi(a', a) = 0\}$.  Moreover,
  let $U = \bigcup_{a \in \bar L} U_a$. We develop a polynomial-time algorithm for the problem stated in the theorem as follows: it returns
  ``{\no}'' if
  $S \not\subseteq f^{(s, 2)}(\varphi, N \setminus U)$ or $|U| > k$,
  and otherwise returns ``{\yes}''.

  The correctness of the algorithm is shown based on the following
  observations.  According to the {\conr} $f^{(s,2)}$,
  $a \in \bar L$ is socially qualified if there is no
  individual $a'\neq a$ such that $\varphi(a', a)=0$. Therefore, in
  order to make $a \in \bar L$ socially qualified, all individuals
  $a' \in N\setminus S$ with $\varphi(a', a) =0$ have to be
  deleted. This directly implies that all individuals in~$U$, as
  defined above, have to be deleted.

  Now let us consider $f^{(s,2)} (\varphi, N \setminus U)$. Suppose
  $S \not\subseteq f^{(s,2)} (\varphi, N \setminus U)$, and let
  $a \in S \setminus f^{(s,2)} (\varphi, N \setminus U)$. We distinguish
  between the following two cases.
  \begin{description}
  \item[Case $a \in L$:] According to the {\conr} $f^{(s,2)}$,
    there are at most $s-1$ individuals $a' \in N\setminus U$ such
    that $\varphi(a', a)=1$.  Since {\excluding} individuals does not
    increase the number of individuals who qualify $a$, the
    individual $a$ cannot be socially qualified after deleting some individuals. Thus, the given
    instance is a {\noins}.
  \item[Case $a \in \bar L$:] In this case, there is an individual
    $a' \in S$ such that $a' \neq a$ and $\varphi(a', a) =0$. Since we
    cannot delete individuals in $S$ due to the definition of the
    problem, individual $a$ cannot be socially qualified. Thus, the
    given instance is a {\noins}.
  \end{description}
  Due to the above analysis, if
  $S \not\subseteq f^{(s,2)}(\varphi, N \setminus U)$, we can safely
  return ``{\no}''. Since we are allowed to delete at most $k$
  individuals, and according to the above analysis all individuals in
  $U$ must be deleted, if $|U|>k$, we can safely return ``{\no}''
  too. On the other hand, if
  $S \subseteq f^{(s,2)}(\varphi, N \setminus U)$ and $|U|\leq k$,
  $U$~itself is an evidence for answering ``{\yes}''.

  Finally, observe that the construction of the set $U$, and the decisions
  of $S \subseteq f^{(s,2)}(\varphi, N \setminus U)$ and $|U|\leq k$
  can be done in ${O}(|N|^2)$ time. This completes the proof.
\end{proof}

\subsection{{\nphns}}
\label{subsec_NP-hard_rules}
In contrast to the polynomial-time solvability of
the {\gcei} problem for consent rules $f^{(s,2)}$ with $s\geq 1$, we prove in
this section that the same problem for consent rules with quotas $s\geq 1$
and $t\geq 3$ becomes {\nph}. Theorem~\ref{thm:GCIIconsentnphard} additionally shows that the {\gcii}
problem for consent rules with quotas $s\geq 2$ and $t\geq 1$ is {\nph},
too. It should be also noted that the instances in our {\nphns}
reductions directly imply that every consent rule $f^{(s,t)}$ with $s\geq 2$
and $t\geq 1$ is susceptible to {\gcii}, and every consent rule $f^{(s,t)}$ with $s\geq 1$ and $t\geq 3$ is susceptible to {\gcei}.

\begin{theorem}
  \label{thm:GCIIconsentnphard}
  {\unsurehidden{The}} {\gcii} {\unsurehidden{problem}} for every {\conr} $f^{(s,t)}$ with $s\geq 2$ and $t\geq 1$, and
 {\unsurehidden{the}} {\gcei} {\unsurehidden{problem}} for every {\conr} $f^{(s,t)}$ with $s\geq 1$ and $t\geq 3$ are {\nph}.
\end{theorem}

\begin{proof}
  We prove the theorem by reductions from the RX3C problem.  Let's
  first consider the {\gcii} problem for {\conr}s $f^{(2,t)}$ with $t\geq 1$.
  Given an instance $\mathcal{I}=(X,\mathcal{C})$ of RX3C
  with $|X| = 3\mathxxxcssize$, we create an instance
  ${\mathgcaiins}=(N, \varphi, S, T, k)$ of {\gcii} for $f^{(s,t)}$ as follows.

  There are $|X| + |\mathcal{C}|$ individuals in
  $N = \{a_x \mid x \in X\}\cup \{a_c \mid c \in \mathcal{C}\}$. The
  first $|X|$ individuals $\{a_x \mid x \in X\}$ one-to-one correspond
  to the elements in $X$, and the last $|\mathcal{C}|$ individuals
  $\{a_c \mid c \in \mathcal{C}\}$ one-to-one correspond to elements
  in $\mathcal{C}$.  We define $S =T = \{a_x \in N \mid x \in X\}$.
  In addition, we set $k=\mathxxxcssize$.  Now we define the profile
  $\varphi$.
  \begin{enumerate}
  \item For each $x, x' \in X$,
    $\varphi(a_x,a_{x'})=1$ if and only if $x=x'$.
  \item For each $x \in X$
    and for each $c \in \mathcal{C}$, $\varphi(a_c, a_x)=1$ if and
    only if $x\in c$.
  \item For each $c, c' \in \mathcal{C}$,
    $\varphi(a_c,a_{c'}) = 0$.
  \end{enumerate}
  For the proof, the values of $\varphi(a_x, a_{c})$ where $x\in X$ and $c\in \mathcal{C}$ are not
  essential. Obviously, the construction of ${\mathgcaiins}$ can be
  done in polynomial time.

Now we prove the correctness of the reduction, i.e., we show that
${\mathxxxcins}$ is a {\yesins} of RX3C if and only if
${\mathgcaiins}$ is a {\yesins} of {\gcii}.

$(\Rightarrow:)$ Suppose ${\mathxxxcins}$ is a {\yesins} for RX3C,
  and let $\mathcal{C}'\subseteq \mathcal{C}$ be an exact 3-set cover, i.e., $|\mathcal{C}'|=\mathxxxcssize$ and for every
  $x\in X$ there exists a $c\in \mathcal{C}'$ such that~$x\in c$.
  Let $U=\{a_c \in N \mid c\in \mathcal{C}'\}$.  Then, according to the definition of $\varphi$, for each
  $a_x \in S$, there exists an $a_c\in U$ such that
  $\varphi(a_c, a_x) = 1$. Moreover, each $a_x \in S$ qualifies
  herself (i.e., $\varphi(a_x,a_x) = 1$). Therefore, according to
  the definition of the {\conr} $f^{(2,t)}$,
  $a_x\in f^{(2,t)}(\varphi, T\cup U)$ for every $a_x\in S$, i.e.,
  $S \subseteq f^{(2,t)}(\varphi, T\cup U)$.  By definition, we have
  $|U| = |\mathcal{C}'|=k = \mathxxxcssize$.  Therefore,
  ${\mathgcaiins}$ is a {\yesins} for
  {\gcii}.

$(\Leftarrow:)$ Suppose ${\mathgcaiins}$ is a {\yesins} for {\gcii}, and let $U\subseteq N \setminus T$ be a set of
  individuals such that $|U|\leq k = \mathxxxcssize$ and
  $S \subseteq f^{(2,t)}(\varphi, T\cup U)$.  From
  $S \subseteq f^{(2,t)}(\varphi, T\cup U)$ and, for all
  $a_x, a_{x'} \in S = T$, $\varphi(a_x,a_{x'})=1$ if and only if
  $x = x'$, it follows that, for each $a_x\in S$, there is an
  $a_c\in U$ such that $\varphi(a_c, a_x)=1$.  Then, according to the
  definition of the profile $\varphi$, for each $x \in X$, there
  exists $c\in \mathcal{C}$ such that $a_c\in U$ and
  $x \in c$.  This implies that
  $\mathcal{C}' = \{c \in \mathcal{C} \mid a_{c}\in U\}$ is
  an exact 3-set cover of $\mathcal{I}$. Thus, $\mathcal{I}$ is
  a {\yesins}.

The {\nphns} reduction for the problem {\gcii} for any {\conr} $f^{(s,t)}$ with $s>2$ and $t\geq 1$ can be
adapted from the above reduction.  Precisely, we introduce
further $s-2$ dummy individuals in $T$, and let all these dummy
individuals qualify every individual in $S=\{a_x\in N\mid x\in X\}$.
The opinions of a dummy individual over any other individual in $N$
and the other way around do not matter in the proof, and thus can be set arbitrarily.
Now for each individual $a_x\in S$, there are exactly $s-1$ individuals in $T$ who qualify
$a_x$.  Moreover, in order to make each $a_x\in S$ socially
qualified, we need one more individual in $N\setminus T$ who qualifies $a_x$.
\bigskip

Now let's consider {\unsurehidden{the}} {\gcei} {\unsurehidden{problem}} for {\conr}s $f^{(s,t)}$ with $s\geq 1$ and $t\geq 3$.  We first
consider the case $t=3$.  The reduction for this problem is similar to the
above reduction for {\unsurehidden{the}} {\gcii} {\unsurehidden{problem}} for {\conr}s $f^{(2,t)}$ with $t\geq 1$ with the following differences.

\begin{enumerate}
\item There is no $T$ in this reduction; 
  but keeping $S=\{a_x\in N\mid x\in X\}$.

\item The values of $\varphi(a,b)$ for every $a,b\in N$
  are reversed. That is, we have $\varphi(a,b)=1$ in the current reduction
  if and only if $\varphi(a,b)=0$ in the above reduction for {\gcii}.

\item $k=2{\mathxxxcssize}$.
\end{enumerate}

Now we prove the correctness
of the reduction.

$(\Rightarrow:)$ Suppose that there is an exact 3-set cover
$\mathcal{C}'\subset \mathcal{C}$ for ${\mathxxxcins}$, i.e., $|\mathcal{C}'|=\mathxxxcssize$ and for every
  $x\in X$ there exists exactly one $c\in \mathcal{C}'$ such that~$x\in c$. Let
$U=\{a_c \mid c\in \mathcal{C}\setminus
\mathcal{C}'\}$ and $U'=\{a_c \mid c\in
\mathcal{C}'\}$.  Clearly, $S\cap U=\emptyset$. Moreover,
$N\setminus U=S\cup U'$.
Let $a_x$ be an individual in $S$ where $x\in X$.
Then, according to the construction, there is exactly one
$a_c\in U'$ such that $\varphi(a_c,a_x)=0$. Since $\varphi(a_{x'},a_{x})=1$ for all $a_{x'}\in S\setminus \{a_x\}$,
according to the {\conr} $f^{(s,3)}$, $a_x\in f^{(s,3)}(\varphi,N\setminus U)$.
Since this holds for every $a_x\in S$, we can conclude that $S\subseteq f^{(s,3)}(\varphi,N\setminus U)$.

$(\Leftarrow:)$ Suppose that there is a
$U\subseteq N\setminus S$ such that $|U|\leq 2\mathxxxcssize$
and \[S\subseteq f^{(s,3)}(\varphi,N\setminus U).\] Let $U'=N\setminus
(S\cup U)$ and $\mathcal{C}'=\{c\in \mathcal{C}\mid a_c\in U'\}$.
Thus, $N\setminus U=S\cup U'$. Due to the fact that $\varphi(a_x,a_x)=0$ for
every $a_x\in S$ where $x\in X$ and the definition of $\varphi$, it holds that for every $a_x\in S$, there
is at most one $a_c\in U'$ such that $\varphi(a_c,a_x)=0$ and $x\in c$.
Due to the construction, every individual $a_c\in U'$ disqualifies exactly 3 individuals in $S$. Then, from $|S|=3\mathxxxcssize$ it follows that $|U'|\leq \mathxxxcssize$.
In addition, from $|U|\leq 2\mathxxxcssize$, we obtain that $|U'|=3\mathxxxcssize-|U|\geq \mathxxxcssize$. Hence, it must be that $|U'|=k$. It follows that  every individual $a_x$ is disqualified by exactly one individual $a_c\in U'$ such that $x\in X, c\in \mathcal{C}'$, and $x\in c$. This implies that the subcollection  $\mathcal{C}'$ corresponding to $U'$ is an exact 3-set cover of $\mathcal{I}$.

The proof of {\nphns} of the problem for any {\conr} $f^{(s,t)}$ with $s\geq 1$ and $t>3$ can be adapted from the above
reduction by introducing some dummy individuals. In particular, we
introduce further $t-3$ individuals in $S$. Let $S'$ denote the set
of the $t-3$ dummy individuals. Thus, $S=\{a_x\in N\mid x\in X\}\cup S'$.
We want each dummy individual in $S'$ to be a robust socially qualified individual, that is,
every $d\in S'$ is socially qualified regardless of which
individuals (at most $k=2\mathxxxcssize$) would be deleted. To this end, for every
$d\in S'$, we let $d$ disqualify herself, and let all the other individuals
qualify $d$.
We set $\varphi(d,a_x)=0$ for every
$d\in S'$ and $a_x$ where $x\in X$. Thus, for every $a_x\in S$ where $x\in X$,
there are in total $t+1$ individuals in $N$ who disqualify $a_x$.
The other entries in the profile
not defined above can be set arbitrarily.
In order to make each $a_x\in S$
where $x\in X$ socially qualified, we need to
delete exactly two individuals in $N\setminus S$ who disqualify
$a_x$. This happens if and only if there is an exact 3-set cover for $\mathcal{I}$, as we
discussed in the proof for the {\conr} $f^{(s,3)}$.
\end{proof}

Even though consent rules $f^{(s,t)}$ with $s\geq 2$ (resp.\ $f^{(s,t)}$ with $t\geq 3$) are susceptible to {\gcii} (resp.\ {\gcei}),
Theorem~\ref{thm:GCIIconsentnphard} reveals that it is a computationally hard task for a designer to successfully control a group
identification procedure in these cases by {\including} (resp.\ {\excluding}) individuals.

Let us now turn
to the {\gcpi} problem. We have shown in Theorems~\ref{thm:liberalimmune} and~\ref{thm:consentimmune} that every consent rule
$f^{(s,1)}$ with $s\geq 1$ is immune to this control type. In order to
show that consent rules with $t>1$ are susceptible to {\gcpi},
consider an instance $\left(N,\varphi ,S\right) $ where $t\geq 2$%
, $N=\left\{ a_{1},a_{2},\ldots ,a_{t+1}\right\} $, $\varphi (a_{i},a_{j})=0$
for every $i,j\in \left\{ 1,2,\ldots ,t+1\right\} $, i.e., everyone disqualifies everyone, and $S=\left\{
a_{1}\right\} $. Clearly, no individual is socially qualified, i.e., $f^{(s,t)}(\varphi ,N)=\emptyset $. Consider now
the partition $(U=\{a_1\},N\setminus \{a_1\})$ of $N$ and note that $f^{(s,t)}(\varphi
,U)=S$. Moreover, for every individual $a_{i}\in N\setminus U$, at least $t$
individuals in $N\setminus U$ disqualify $a_{i}$ and thus $f^{(s,t)}(\varphi ,N\setminus U)=\emptyset$. We have then $S=\left\{
a_{1}\right\} =f^{(s,t)}(\varphi ,f^{(s,t)}(\varphi ,U)\cup
f^{(s,t)}(\varphi ,N\setminus U))$, showing the susceptibility of the consent
rule to {\gcpi}. Our next result reveals that, in fact,
manipulation by partitioning the set of individuals is {\nph}, provided
that the social rule is $f^{(s,t)}$ with $t\geq 2$.

\begin{theorem}
\label{thm_GCPI_consent_s_t_NP_Hard}
{\unsurehidden{The}} {\gcpi} {\unsurehidden{problem}} for every consent rule $f^{(s,t)}$ with $s\geq 1$ and $t\geq 2$ is {\nph}.
\end{theorem}

\begin{proof}
We prove the {\nphns} of the problem stated in the theorem by a reduction from the {\threesat} problem. We first consider {\unsurehidden{the}} {\gcpi} {\unsurehidden{problem}} for {\conr}s $f^{(s,2)}$ with $s\geq 1$. Later, we extend the reduction to all consent rules $f^{(s,t)}$ with $s\geq 1$ and $t\geq 3$.

Let $(X,C)$ be an instance of the {\threesat} problem, where $X$ is the set of Boolean variables and $C$ is the set of clauses each consisting of three literals. Moreover, let $m$ and $n$ be the numbers of variables and clauses, respectively, i.e., $m=|X|$ and $n=|C|$. We construct an instance $\mathcal{E}=(N,\varphi, S)$ of {\gcpi} for $f^{(s,2)}$ as follows.

There are in total $%
2m+n+1$ individuals in $N$. In particular, for each variable $x\in X$ and
each clause $c\in C$, we create the individuals $a(x,1)$, $a(x,2)$, and $%
a(c) $, respectively. Moreover, we create one individual $a(C)$ for $C$. We
set $S=\left\{ a(x,1)\mid x\in X\right\} \cup \left\{ a(C)\right\} $ and
define the profile $\varphi $ as follows.

\begin{enumerate}
\item For each $a\in N$, $\varphi(a,a)=0$.

\item For each $x\in X$, $\varphi(a(x,2),a(x,1))=0$.

\item For each $c\in C$, $\varphi(a(c),a(C))=0$.

\item For each $x\in X$, $\varphi(a(C),a(x,2))=0$.

\item For each $x\in X$ and $c\in C$, $\varphi(a(c),a(x,2))=0$.

\item For each $c\in C$ and every variable $x$ involved in $c$, $\varphi(a(x,2),a(c))=0$ if $x\in c$ and $\varphi(a(x,1),a(c))=0$ if $\bar{x}\in c$.

\item For every two $a,a'\in N$ such that $\varphi(a,a')$ is not defined above, $\varphi(a,a')=1$.
\end{enumerate}

\begin{table}
\caption{This table summarizes, for each individual $a\in N$, the set of individuals qualifying $a$ and the set of individuals disqualifying $a$, according to the profile $\varphi$ in the proof of Theorem~\ref{thm_GCPI_consent_s_t_NP_Hard}.}
\begin{center}
\begin{tabular}{lll}\toprule
         &   qualified by                            &  \large{disqualified by} \\ \midrule

{$a(x,1)$}   &   $N\setminus \{a(x,1),a(x,2)\}$          &  \large{$a(x,1),a(x,2)$} \\ \midrule

\multirow{2}{*}{$a(x,2)$}   &    \multirow{2}{*}{$N\setminus (\{a(x,2),a(C)\}\cup \{a(c)\mid c\in C\})$}       &   \large{$a(x,2), a(C),$ and} \\
                     &                     &  \large{$a(c)$ for each $c\in C$}\\ \midrule

\multirow{2}{*}{$a(C)$}   &    \multirow{2}{*}{$N\setminus (\{a(c)\mid c\in C\}\cup \{a(C)\})$}           &   \large{$a(C)$ and} \\
 & &  \large{$a(c)$ for every $c\in C$}\\ \midrule

\multirow{3}{*}{$a(c)$}   &   \multirow{3}{*}{$N\setminus (\{a(x,2)\mid x\in c\} \cup \{a(x,1)\mid \bar{x}\in c\}\cup \{a(c)\})$}                     &   \large{$a(c),$}\\
                        &                         &  \large{$a(x,2)$ for every $x\in c,$} \\
                        &                        &  \large{$a(x,1)$ for every $\bar{x}\in c$}\\ \bottomrule
\end{tabular}
\end{center}
\label{table_GCPI_consent}
\end{table}

Now we prove that $(X,C)$ is a {\yesins} if and only if $\mathcal{E}$ is a {\yesins}. Table~\ref{table_GCPI_consent} is helpful for the reader to check the following arguments.

$(\Rightarrow:)$ Assume that there is a truth assignment $\varrho: X\rightarrow \{0,1\}$. Then, we find a $U\subseteq N$ as follows. First, in $U$ we include the individual $a(C)$ and exclude all individuals in $\{a(c)\mid c\in C\}$, i.e., $a(C)\in U$ and $\{a(c)\mid c\in C\}\subseteq N\setminus U$. In addition, for each $x\in X$, $U$ includes exactly one of $\{a(x,1),a(x,2)\}$, depending on the value of $\varrho(x)$. In particular, for every $x\in X$, $a(x,1)\in U$ and $a(x,2)\in N\setminus U$ if $\varrho(x)=1$; and $a(x,2)\in U$ and $a(x,1)\in N\setminus U$; otherwise. Now let's consider the subprofiles $f^{(s,2)}(\varphi, U)$ and $f^{(s,2)}(\varphi, N\setminus U)$. Since the only individual in $U$ who disqualifies $a(C)$ is $a(C)$ herself, it holds that $a(C)\in f^{(s,2)}(\varphi,U)$. Let $x$ be any variable in $X$. Due to the above definition of $U$, the individuals $a(x,2)$
and $a(x,1)$ are not included in the same element of the partition $
(U,N\setminus U)$ of $N$. Since the only individuals who disqualify $a(x,1)$ are $a(x,1)$ and $a(x,2)$, it holds that $a(x,1)$ survives the first stage of selection, i.e., $a(x,1)\in f^{(s,2)}(\varphi,U)$ if $a(x,1)\in U$ and $a(x,1)\in f^{(s,2)}(\varphi, N\setminus U)$ if $a(x,1)\in N\setminus U$. On the other hand, since $a(C)$ and all individuals in $\{a(c)\mid c\in C\}$ disqualify every individual in $\{a(x,2)\mid x\in X\}$, none of $\{a(x,2)\mid x\in X\}$ survives the first stage of selection, i.e., for every $x\in X$ it holds that $a(x,2)\not\in f^{(s,2)}(\varphi,W)$ where $W\in \{U,N\setminus U\}$ and $a(x,2)\in W$. Now we consider the individuals corresponding to the clauses. Let $c\in C$ be a clause. Since $c$ is satisfied under $\varrho$, there is either an $x\in c$ such that $\varrho(x)=1$, or a $\bar{x}\in c$ such that $\varrho(x)=0$. In the former case, we have $a(x,2)\in N\setminus U, \varphi(a(x,2),a(c))=0$, and in the latter case we have
$a(x,1)\in N\setminus U, \varphi(a(x,1),a(c))=0$. Hence, both cases lead $a(c)$ to be eliminated in the first stage of selection, i.e., $\{a(c)\mid c\in C\}\cap f^{(s,2)}(\varphi,N\setminus U)=\emptyset$. As a summary, $f^{(s,2)}(\varphi,U)\cup f^{(s,2)}(\varphi,N\setminus U)=S$. Since every individual in $S$ is only disqualified by herself, i.e., $\varphi(a',a)=0$ if and only if $a'=a$ for every $a,a'\in S$, we have that $S=f^{(s,2)}(\varphi,S)$. This completes the proof of this direction.

$(\Leftarrow:)$ Suppose that there is a $U\subseteq N$ such that $S\subseteq f^{(s,2)}(\varphi,f^{(s,2)}(\varphi,U)\cup f^{(s,2)}(\varphi,N\setminus U))$. Due to symmetry, assume that $a(C)\in U$. Since for every $c\in C$, $\varphi(a(c),a(C))=0$, and $a(C)\in S$, it holds that $\{a(c)\mid c\in C\}\subseteq N\setminus U$ (otherwise, $a(C)$ would be eliminated in the subprofile restricted to $U$). Moreover, it holds that $\{a(c)\mid c\in C\}\cap f^{(s,2)}(\varphi,N\setminus U)=\emptyset$. As a result, for every $a(c)$, except herself, there must be at least one other individual in $ N\setminus U$ who disqualifies $a(c)$. Due to the definition of the profile, this means that there is either an $x\in c$ such that $a(x,2)\in N\setminus U$, or a $\bar{x}\in c$ such that $a(x,1)\in N\setminus U$. Since for every $x\in X$ it holds $\varphi(a(x,2),a(x,1))=\varphi(a(x,1),a(x,1))=0$, exactly one of $\{a(x,1),a(x,2)\}$ can be in $N\setminus U$ (otherwise, $a(x,1)$ would be eliminated in the first stage of selection). Hence, given $U$, we can uniquely define a truth assignment $\varrho$ as follows. For every $x\in X$, define $\varrho(x)=1$ if $a(x,2)\in N\setminus U$ and $a(x)=0$ otherwise. Then, due to the above discussion, for every $c\in C$, there is either an $x\in c$ such that $\varrho(x)=1$ or a $\bar{x}\in c$ such that $\varrho(x)=0$. Therefore, every clause is satisfied under the truth assignment $\varrho$. This completes the proof of this direction.

Now, we explain how to extend the above reduction for each consent rule $f^{(s,t)}$ with $t\geq 3$. Assume that $|C|\geq t-1$ (if this is not the case, we can duplicate any arbitrary clause to make the inequality hold). In addition to the individuals defined in the above reduction, we further create $2t-4$ dummy individuals $a_1^1,\dots,a_1^{t-2}, a_2^1,\dots,a_2^{t-2}$, and include all individuals $a_1^1,\dots,a_1^{t-2}$ in $S$. Let $A_1=\{a_1^1,\dots,a_1^{t-2}\}$ and $A_2=\{a_2^1,\dots,a_2^{t-2}\}$. So we now have $S=\{a(x,1)\mid x\in X\}\cup a(C)\cup A_1$, and \[N=S\cup A_2\cup \{a(x,2)\mid x\in X\}\cup \{a(c)\mid c\in C\}.\] Every individual in $A_1\cup A_2$ is disqualified by all individuals in $A_1\cup A_2$. In addition, all individuals in $N\setminus (A_1\cup A_2)$ disqualify every individual in $A_2$, and all individuals in $N\setminus (A_1\cup A_2\cup \{a(c)\mid c\in C\})$ qualify every individual in $A_1$. Furthermore, all individuals in $\{a(c)\mid c\in C\}$ disqualify all individuals in $A_1$. Finally, all individuals in $A_1\cup A_2$ disqualify all individuals not in $A_1\cup A_2$. The subprofile restricted to individuals not in $A_1\cup A_2$ remains unchanged. Observe that, by defining so, to make all individuals in $A_1$ socially qualified, it has to be the case that for every solution $U$ the number of dummy individuals included in $U$ and $ N\setminus U$ should be the same (and equal to $t-2$). Assume for the sake of contradiction that this is not the case. Let $U\subseteq N$ be a solution. Due to symmetry, we assume that $|U\cap (A_1\cup A_2)|\geq t-1$. Apparently, $U$ includes at least one individual $a_1^i\in S\cap A_1$. Moreover, all individuals in $\{a(x,1)\mid x\in X\}\cup \{a(C)\}$, all of which are in $S$, must be partitioned into the set $N\setminus U$, since otherwise all of them will be eliminated in the first stage of selection (i.e., for every $a\in \{a(x,1)\mid x\in X\}\cup \{a(C)\}$ it holds that $a\not\in f^{(s,t)}(\varphi,W)$ where $W\in \{U,N\setminus U\}$ and $a\in W$). As a result, at most $t-2$ individuals in $\{a(c)\mid c\in C\}$ can be included in $N\setminus U$, since otherwise the individual $a(C)$ will be eliminated, i.e., $a(C)\not\in f^{(s,t)}(\varphi,N\setminus U)$. As $|C|\geq t-1$, there will be at least one individual $a(c)$ with $c\in C$ in the set $U$. However, the individual $a(c)$ together with all other individuals in $U\cap (A_1\cup A_2)$ will make $a_1^i$ be eliminated, contradicting that $U$ is a solution. The observation follows. The discussion for the observation also implies that for every solution $U\subseteq N$, either $A_1\subseteq U$ or $A_1\subseteq N\setminus U$, i.e., all individuals in $A_1$ must be included in the same element of the partition $(U,N\setminus U)$ of $N$.
Now, one can check that from every solution $U$ of the instance constructed above for the consent rule $f^{(s,2)}$, we can get a solution for the instance constructed for the consent rule $f^{(s,t)}$ with $s\geq 1$ and $t\geq 3$ by {\including} all individuals $A_1$ in $U$ or $N\setminus U$ (if $a(C)\in U$ then $A_1\subseteq U$; otherwise $A_1\subseteq N\setminus U$), and vice versa. This completes the proof.
\end{proof}

\section{Procedural rules}\label{sec_procedure_complexity}
Each of the procedural rules for group identification introduced in Section~\ref{sec_notation} expands an initial set of socially qualified individuals by adding new individuals. Hence, the very definition of these rules invites us to start the analysis of whether they are susceptible to {\gcii}. We answer the question in the affirmative. As a matter of fact, we show that {\gcii} for $f^{LSR}$ and {\gcii} for $f^{CSR}$ are both {\nph}, which directly implies the susceptibility of the two procedural rules to {\gcii}.

\begin{theorem}\label{thm:GCIICSRnphard}
The {\gcii} problem for $f^{LSR}$ and the {\gcii} problem for $f^{CSR}$ are {\nph}.
\end{theorem}

\begin{proof}
  We prove the theorem by reductions from the RX3C problem.  Let's
  first consider the social rule $f^{LSR}$. Given an instance
  $\mathxxxcins=(X, \mathcal{C})$ of RX3C with $|X|=3\kappa$, we create an
  instance ${\mathgcaiins}=(N,\varphi,S,T,k)$ of {\gcii} for $f^{LSR}$ as follows.

  The definitions of $N,S,T$ and $k$ are the same as in
  the {\nphns} reduction for {\gcii} for {\conr}s $f^{(2,t)}$ with $t\geq 1$ in
  Theorem~\ref{thm:GCIIconsentnphard}. That is, $N=\{a_x \mid x \in X\}\cup \{a_c \mid c \in \mathcal{C}\}$ is a set of $|X| + |\mathcal{C}|$ individuals, with the first $|X|$ individuals $\{a_x \mid x \in X\}$ one-to-one corresponding
  to the elements in $X$, and the last $|\mathcal{C}|$ individuals
  $\{a_c \mid c \in \mathcal{C}\}$ one-to-one corresponding to elements
  in $\mathcal{C}$.  In addition, $S =T = \{a_x \in N \mid x \in X\}$, and $k=\mathxxxcssize$.
The profile $\varphi$ is defined as follows.

  \begin{enumerate}
  \item For each $x,x'\in X$, $\varphi(a_x,a_{x'})=0$.

  \item For each $x\in X$ and each $c\in \mathcal{C}$, $\varphi(a_x,a_c)=0$.

  \item For each $c,c'\in \mathcal{C}$, $\varphi(a_c,a_{c'})=1$ if and only if $c=c'$.

  \item For each $x\in X$ and each $c\in \mathcal{C}$, $\varphi(a_c,a_x)=1$ if and only if $x\in c$.
  \end{enumerate}

 Now we prove the correctness of the reduction.

  $(\Rightarrow:)$ Suppose that there is an exact 3-set cover
$\mathcal{C}'\subset \mathcal{C}$ for ${\mathxxxcins}$, i.e., $|\mathcal{C}'|=k$ and for every
  $x\in X$ there exists exactly one $c\in \mathcal{C}'$ such that~$x\in c$. Let
$U=\{a_c \mid c\in \mathcal{C}'\}$.
  According to the definition of $\varphi$, it holds that
  $U\subseteq f^{LSR}(\varphi, T\cup U)$.
  Moreover, for every $a_x\in S$ with $x\in X$,
  there is an $a_c\in U$ such that $\varphi(a_c,a_x)=1$ and $x\in c$.
  Since $U\subseteq f^{LSR}(\varphi,T\cup U)$, according to the
  definition of the social rule $f^{LSR}$, it holds that
  $a_x\in f^{LSR}(\varphi, T\cup U)$ for every $a_x\in S$.
  Therefore, ${\mathgcaiins}$ is a {\yesins} since
  it has a solution $U$.

  $(\Leftarrow:)$ Suppose that there is a $U\subseteq N\setminus T$ such
  that $|U|\leq k$ and $S=T\subseteq f^{LSR}(\varphi, T\cup
  U)$. Let $\mathcal{C}'=\{c\in \mathcal{C}\mid a_c\in U\}$.
  According to the definition of $\varphi$,
  $f^{LSR}(\varphi,T)=\emptyset$. Moreover, every $a_x\in S$ with $x\in X$
  disqualifies all individuals in $N$, and every
  $a_c\in N\setminus T$ qualifies herself.
  As a result, for every $a_x\in S$ with $x\in X$, there must be at least one
  $a_c\in U$ with $c\in \mathcal{C}'$ such that
  $\varphi(a_c,a_x)=1$. According to the definition of
  $\varphi$, this implies that for every $x\in X$, there is
  at least one $c\in \mathcal{C}'$ such that $x\in c$. Since $|\mathcal{C}'|=|U|\leq k=\kappa$,
  this implies that $|\mathcal{C}'|=k$ and, more precisely,
  $\mathcal{C}'$ is an exact 3-set cover of ${\mathxxxcins}$.

  \bigskip

  Now let's consider {\unsurehidden{the}} {\gcii} {\unsurehidden{problem}} for $f^{CSR}$.
  Again, the definitions of $N,S,T$ and $k$ are the same as in
  the {\nphns} reduction for {\unsurehidden{the}} {\gcii} {\unsurehidden{problem}} for {\conr}s $f^{(2,t)}$ with $t\geq 1$ in
  Theorem~\ref{thm:GCIIconsentnphard}. The profile $\varphi$ is
  defined as follows.

  \begin{enumerate}
  \item For each $x,x'\in X$, $\varphi(a_x,a_{x'})=0$.

  \item For each $x\in X$ and each $c\in \mathcal{C}$, $\varphi(a_x,a_c)=1$.

  \item For each $c,c'\in \mathcal{C}$, $\varphi(a_c,a_{c'})=1$.

  \item For each $x\in X$ and each $c\in \mathcal{C}$, $\varphi(a_c,a_x)=1$ if and only if $x\in c$.
  \end{enumerate}

  Now we prove the correctness of the reduction.

  $(\Rightarrow:)$ Suppose that there is a
$\mathcal{C}'\subset \mathcal{C}$ such that $|\mathcal{C}'|=k$ and for every
  $x\in X$ there exists exactly one $c\in \mathcal{C}'$ such that~$x\in c$. Let
$U=\{a_c \mid c\in \mathcal{C}'\}$. Clearly, $|U|=|\mathcal{C}'|=k$.
 Observe that $U\subseteq f^{CSR}(\varphi,T\cup U)$.
  Then, according to the definition of $\varphi$,
  it holds that for every $a_x\in S$ with $x\in X$, there is an $a_c\in U$
  such that $x\in c$ and $\varphi(a_c,a_x)=1$. This implies that
  $a_x\in f^{CSR}(\varphi,T\cup U)$ for every $a_x\in S$.
  Thus, ${\mathgcaiins}$ is a {\yesins} since it has a solution $U$.

  $(\Leftarrow:)$ Suppose that there is a subset $U\subseteq N\setminus T$ such
  that $|U|\leq k$ and $S=T\subseteq f^{CSR}(\varphi,T\cup U)$.
  Let $\mathcal{C}'=\{c\in \mathcal{C}\mid a_c\in U\}$.
  According to the definition of $\varphi$,
  $f^{CSR}(\varphi,T)=\emptyset$. Moreover, every individual in $S$
  disqualifies every individual in $S$. Furthermore, every individual
  in $N\setminus T$ is qualified by all individuals in $N$.
  Therefore, for every $a_x\in S$ with $x\in X$, there must be at least one
  $a_c\in U$ such that $\varphi(a_c,a_x)=1$. According to the
  definition of $\varphi$, this implies that for every
  $x\in X$ there is at least one $c\in \mathcal{C}'$ such that
  $x\in c$.  Since $|\mathcal{C}'|=|U|\leq k=\kappa$,
  this implies that $|\mathcal{C}'|=k$ and, more precisely,
  $\mathcal{C}'$ is an exact 3-set cover of ${\mathxxxcins}$.
\end{proof}

In contrast to the susceptibility of the procedural rules to {\gcii}, we show next that $f^{LSR}$ turns out to be immune to the other two group control types. Intuitively, an individual is not socially qualified if there are not enough individuals qualifying her. Hence, deleting some individuals cannot make such an individual socially qualified; it in fact can only make the situation worse for the  individual.

\begin{theorem}
  \label{thm:GCEIGCPICSRLSRimmune}
  The social rule $f^{LSR}$ is immune to {\gcei} and {\gcpi}.
\end{theorem}

\begin{proof}
  According to the definition of $f^{LSR}$, an individual $a\in  N$
  is a socially qualified with respect to $f^{LSR}$
  if and only if
  there is a sequence of individuals $a_0,a_1,\dots,a_t$ such that
  (1) $a_t=a$;
  (2) $\varphi(a_0,a_0)=1$, i.e., $a_0$ is in the initial set of socially qualified individuals; and
  (3) $\varphi(a_i, a_{i+1})=1$ for every $0\leq i\leq t-1$.
Clearly, if there is no such path for an individual $a$, then after deleting some individuals such a path still does not exist for $a$. Hence, $f^{LSR}$ is immune to {\gcei}. Moreover, for any partition $N_1$ and $N_2$ of $N$ such a path also does not exist in the subprofile restricted to both $N_1$ and $N_2$, if it does not exist in the overall profile. Hence, $f^{LSR}$ is immune to {\gcpi} too.
\end{proof}

In contrast to the immunity of $f^{LSR}$ to {\gcei} and {\gcpi}, we show that $f^{CSR}$ is susceptible to {\gcei} and {\gcpi}. Consider the instance where there are three individuals $a,b,c$ and $S=\{a,b\}$. In addition, $\varphi(a,a)=\varphi(a,b)=\varphi(b,b)=\varphi(b,a)=1$ and $\varphi(c,a)=\varphi(c,b)=\varphi(c,c)=0$. Then, initially no one is qualified by all individuals. Hence, there are no socially qualified individuals. However, after deleting $c$ both $a$ and $b$ become socially qualified individuals. In addition, if we partition the individuals as $(\{a,b\},\{c\})$, both $a$ and $b$ become socially qualified individuals too. Next, we prove that {\gcei} for $f^{CSR}$ is polynomial-time solvable.

\begin{theorem}
\label{thm-liberal-rule-polynomial-ccdi}
{\gcei} for $f^{CSR}$ is polynomial time solvable.
\end{theorem}

\begin{proof}
To prove the theorem, we develop a polynomial-time algorithm. For an individual $a\in N$, let $D(a)$ be the set of individuals disqualifying $a$, i.e., $D(a)=\{a'\in N \mid \varphi(a',a)=0\}$. The algorithm is as follows: return ``{\yes}'' if and only if there is an individual $a\in N$ such that $|D(a)|\leq k$ and $S\subseteq f^{CSR}(\varphi, N\setminus D(a))$. Clearly, the algorithm can be implemented in polynomial time. It remains to prove its correctness. Obviously, if the algorithm returns ``{\yes}'', the given instance must be a {\yesins}. Assume now that the given instance is a {\yesins}. Let $U\subseteq N\setminus S$ be a solution of the given instance, i.e., $|U|\leq k$ and $S\subseteq f^{CSR}(\varphi, N\setminus U)$.
Let $K_0^{CSR}$ be the initial set of socially qualified individuals with respect to $\varphi$ and $N\setminus U$.
Clearly, $\bigcup_{b'\in K_0^{CSR}}D(b')\subseteq U$.
Let $b$ be any arbitrary individual in $K_0^{CSR}$. We claim that $D(b)$ is also a solution. To check this, first observe that all individuals in $K_0^{CSR}$ qualify all individuals in $K_0^{CSR}$, i.e., $\varphi(c,c')=1$ for all $c,c'\in K_0^{CSR}$.
Moreover, according to the definition of $f^{CSR}$, all individuals in $K_0^{CSR}$ are socially qualified with respect to $\varphi$ and $N\setminus D(b)$ (hint: $b$ is socially qualified with respect to $N\setminus D(b)$ and $b$ qualifies everyone in $K_0^{CSR}$ as discussed above). This implies that all individuals in $S$ are socially qualified with respect to $\varphi$ and $N\setminus D(b)$. As $D(b)\subseteq U$ and $|U|\leq k$, it holds that $|D(b)|\leq k$. Hence, the individual $b$ is a witness (i.e., $b$ is the individual $a$ as described in the algorithm) that the algorithm returns ``{\yes}''. This completes the proof.
\end{proof}

\section{Bounded group size and parameterized complexity}
\label{sec_fpt}
We have shown in Theorems~\ref{thm:GCIIconsentnphard} and~\ref{thm_GCPI_consent_s_t_NP_Hard} that the {\gcii}, {\gcei}, and {\gcpi} problems for {\conr}s $f^{(s,t)}$ are {\nph} when either $s$ or $t$ exceeds some constant. Hence, there are no polynomial-time algorithms for these problems unless {\poly}={\np}.
In this section, we investigate how the size of the group $S$ of people to be made
socially qualified affects the complexity of the problems studied in the previous section.
In particular, we study the {\gcii, \gcei}, and {\gcpi} problems from the parameterized complexity point of view, with respect to the size of $S$.

Parameterized complexity was introduced by Downey and Fellows~\citeasnoun{fellows99} as a tool to deal with hard problems.
A {\it{parameterized problem}} is a language contained in $\Sigma^*\times\Sigma^*$, where $\Sigma$ is a finite alphabet.
The first component is called the {\it{main part}} of the problem and the second component is called the {\it{parameter}}.
In this paper, we consider only positive integer parameters.
A parameterized problem is {\it{fixed-parameter tractable}} ({\fpt}) if it is solvable in $O(f(k)\cdot|I|^{O(1)})$ time, where $I$ is the main part of the instance, $k$ is the parameter, and $f(k)$ is a computable function depending only on $k$. For further discussion on parameterized complexity, we refer to~\citeasnoun{DBLP:books/sp/CyganFKLMPPS15,DBLP:series/txcs/DowneyF13,rolf06}.

We first study the {\gcii} and the {\gcei} problems for consent rules. In particular,
we prove that both the {\gcii} and {\gcei} problems for these rules are {\fpt} with
respect to the size of $S$.
To this end, we give integer linear programming (ILP)
formulations with the number of variables bounded by $2^{|S|}$ for
both problems.
As ILP is {\fpt} with respect to the number of
variables~\cite{Frank1987,kannan87a,lenstra83}, so are the {\gcii} and {\gcei} problems for the consent rules.

\begin{lemma}\cite{Frank1987,kannan87a,lenstra83}
  \label{lenstra} ILP can be solved using
  $O(v^{2.5v+\small{o}(v)}\cdot L)$ arithmetic operations, where $L$ is the number of bits in the input and
  $v$ is the number of variables in ILP.
\end{lemma}

Let us now describe the ILP formulations for both the {\gcii} problem
and the {\gcei} problem for {\conr}s.

\begin{theorem}
\label{thm:GCIIGCEIconsentrulefpt}
  The {\gcii} and {\gcei} problems for every {\conr} $f^{(s,t)}$ are {\fpt} with respect to the size
  of $S$.
\end{theorem}

\begin{proof}
  We prove the theorem by giving ILP formulations for the {\gcii}
  and {\gcei} problems. The number of variables in the formulations is bounded by a function
  of $|S|$.
We first consider the {\gcii} problem.

  Let $(N,\varphi,S,T,k)$ be an instance of {\gcii} for $f^{(s, t)}$.
  Let $\mu=|S|$.
  We say two individuals $a,b\in N$
  have the \textit{same opinion} over $S$, if for every $c\in S$,
  it holds that $\varphi(a,c)=\varphi(b,c)$.
  Hereinafter, let $(a_1,a_2,\dots,a_n)$ be any arbitrary but fixed order of $N$.
  Let $S=\{a_{\lambda{(1)}},a_{\lambda{(2)}},\dots,a_{\lambda{(\mu)}}\}$
  where $1\leq \lambda(i)< \lambda(j)\leq n$ for every $1\leq i< j\leq \mu$.
  For an individual $a_i\in N$ where $1\leq i\leq n$,
  let $\varphi_{(a_i,S)}$ denote the vector $\langle \varphi(a_i,a_{\lambda{(1)}}), \varphi(a_i,a_{\lambda{(2)}}),\dots, \varphi(a_i,a_{\lambda{(\mu)}})\rangle$.

  The ILP formulation for
  the instance is as follows.
  For every $\mu$-dimensional 1-0 vector $\beta$, let
  $N_{\beta}=\{a_i\in N\setminus T \mid \varphi_{(a_i,S)}=\beta\}$ and
  $n_{\beta}=|N_{\beta}|$. We create a variable $x_{\beta}$ for
  every $\mu$-dimensional 1-0 vector $\beta$. Thus,
  there are in total $2^{\mu}$ variables.
  Each variable $x_{\beta}$ indicates how many individuals from $N_{\beta}$ are
  included in the solution $U$. These variables are subject to the
  following restrictions. Let $\mathfrak{V}$ be the set of all ${\mu}$-dimensional 1-0 vectors.

  (1) Since for every $\mu$-dimensional 1-0 vector $\beta$ there are at most
  $n_{\beta}$ individuals $a_i\in N\setminus T$ such that $\varphi_{(a_i,S)}=\beta$, we
  need to ensure that no more than $n_{\beta}$ of these individuals
  are in $U$. Moreover, every variable should be non-negative. Thus,
  every variable $x_{\beta}$ is subject to

  \[0\leq x_{\beta}\leq n_{\beta}.\]

  (2) Since we can add at most $k$ individuals in
  total, it has to be that

  \[\sum_{\beta\in {\mathfrak{V}}}x_{\beta}\leq k.\]

  (3) In order to make every individual in $S$ socially qualified,
  it has to be that

  (3.1) for every $a_{\lambda(i)}\in S$ such that $\varphi(a_{\lambda(i)},a_{\lambda(i)})=1$

  \[\sum_{a_j\in T}\varphi(a_j,a_{\lambda(i)})+\sum_{\beta\in \mathfrak{B}}(\beta[i]\cdot
  x_{\beta})\geq s;~\text{and}\]

  (3.2) for every $a_{\lambda(i)}\in S$ such that $\varphi(a_{\lambda(i)},a_{\lambda(i)})=0$

  \[\sum_{a_j\in T}(1-\varphi{(a_j,a_{\lambda(i)})})+\sum_{\beta\in \mathfrak{B}}((1-\beta[i])\cdot x_{\beta})\leq t-1,\] where $\beta[i]$ is the $i$-th component of $\beta$. The
  inequality (3.1) is to ensure that for every $a_{\lambda(i)}\in S$ who
  qualifies herself there are at least $s$
  individuals in the final profile who qualify $a_{\lambda(i)}$,
  and the inequality (3.2) is to ensure that for
  every individual $a_{\lambda(i)}\in S$ who disqualifies herself there
  are at most $t-1$ individuals in the final profile who disqualify $a_{\lambda(i)}$.

  Now let's consider the {\gcei} problem. Let $(f^{(s,t)},N,\varphi,S,k)$ be a given instance of the {\gcei} problem.
 The ILP formulation for the instance is similar to the one for the {\gcii} problem discussed above.
 Let $(a_1,a_2,\dots,a_n)$,
 $\{a_{\lambda{(1)}},a_{\lambda{(2)}},\dots,a_{\lambda{(\mu)}}\}$, $\varphi_{(a_i,S)}$, and $\mathfrak{V}$ be defined with the same meanings as above.
 For every ${\mu}$-dimensional 1-0 vector $\beta$, let
  $\overline{N}_{\beta}=\{a_{j}\in N\setminus S \mid
  \varphi_{(a_j,S)}=\beta\}$ and
  $\overline{n}_{\beta}=|\overline{N}_{\beta}|$. We create a variable
  $y_{\beta}$ for every $\beta\in \mathfrak{V}$.
  Each variable $y_{\beta}$ indicates how
  many individuals from $\overline{N}_{\beta}$ are deleted. The
  restrictions are as follows.

  (1) For every $\beta\in \mathfrak{V}$ we can delete at most
  $\overline{n}_{\beta}$ individuals in
  $\overline{N}_{\beta}$. Moreover, each variable should be
  non-negative. Thus, for every variable $y_{\beta}$, we have that

  \[0\leq y_{\beta}\leq \bar{n}_{\beta}.\]

  (2) Since we can delete at most $k$ individuals in total, we have that

  \[\sum_{\beta\in \mathfrak{V}}y_{\beta}\leq k.\]

  (3) In order to make every individual in $S$ socially
  qualified, it has to be that

  (3.1) for every $a_{\lambda(i)}\in S$ such that $\varphi{(a_{\lambda(i)},a_{\lambda(i)})}=1$

  \[\sum_{a_j\in N}\varphi{(a_j,a_{\lambda(i)})}-\sum_{\beta\in \mathfrak{V}}(\beta[i]\cdot
  y_{\beta})\geq s;~\text{and}\]

  (3.2) for every $a_{\lambda(i)}\in S$ such that $\varphi{(a_{\lambda(i)},a_{\lambda(i)})}=0$

  \[\sum_{a_j\in N}(1-\varphi(a_j,a_{\lambda(i)}))-\sum_{\beta\in \mathfrak{V}}((1-\beta[i])\cdot y_{\beta})\leq t-1.\]

According to Lemma~\ref{lenstra}, both ILPs shown above are
  solvable in time $O(v^{2.5v+\small{o}(v)}\cdot poly(v\cdot n))$, where
  $v=2^{\mu}$. As a result, both the {\gcii} and {\gcei} problems are {\fpt}
  with respect to $\mu=|S|$.
\end{proof}

Consider now the {\gcpi} problem for consent rules. In contrast to the fixed-parameter tractability of the {\gcii} and {\gcei} problems, we show that the {\gcpi} problem is unlikely to admit an {\fpt}-algorithm. In particular, we prove that the {\gcpi} problem for {\conr}s $f^{(s,2)}$ with $s\geq 3$ remains {\nph} even when $S$ is a singleton. This directly implies that the {\gcpi} problem for consent rules is not {\fpt}~\footnote{In fact, this implies that {\unsurehidden{the}} {\gcpi} {\unsurehidden{problem}} for consent rules is even beyond XP, the class of all parameterized problems which are solvable in $O(|I|^{f(k)})$ time, where $I$ is the main part, $k$ is the parameter and $f$ is a computable function.}.
Our reduction is from the {{LRBDS}} problem which is {\nph} as stated in the following lemma.

\begin{lemma}
\label{lem:multicolorredbluedominatingisnphard}
The {{LRBDS}} problem is {\nph}.
\end{lemma}

The proof for the above lemma is deferred to the Appendix.

\begin{theorem}
\label{thm:GCPIconsentthreetwonphard}
The {\gcpi} problem is {\nph} for {\conr}s $f^{(s,2)}$ with $s\geq 3$, even when $|S|=1$.
\end{theorem}

\begin{proof}
We prove the theorem by a reduction from the LRBDS problem. Let $\mathcal{I}=(G=(R\uplus B,E), \{1,2,\dots,k\})$ be an instance of the LRBDS problem.
Let $s\geq 3$. We create an instance ${\mathGCPIthreetwoins}=(N,\varphi,S)$ of {\gcpi} for $f^{(s,2)}$ as follows.
We create $k+s-2+|B|+|R|$ individuals in total. Let $(R_1,R_2,\dots,R_k)$ be the partition of $R$ with respect to the labels of the vertices. That is, $R_i$ where $1\leq i\leq k$, is the set of vertices in $R$ with label $i$. For each vertex $v\in R_i$ where $1\leq i\leq k$, we create an individual $a_i(v)$. Let $A_i=\{a_i(v)\mid v\in R_i\}$. Moreover, for every vertex $u\in B$, we create an individual $a(u)$. Let $A(B)=\{a(u)\mid u\in B\}$. In addition, we create a set $C=\{c_1,c_2,\dots,c_k\}$ of $k$ individuals where each $c_i, 1\leq i\leq k$, corresponds to the label $i$. Moreover, we create an individual $w$ and set $S=\{w\}$. Finally, we create a set $A_{dummy}=\{d_1,d_2,\dots,d_{s-3}\}$ of $s-3$ dummy individuals. Hence, $N=\bigcup_{1\leq i\leq k}A_i\cup A(B)\cup C\cup A_{dummy}\cup S$. The profile $\varphi$ is defined as follows.

\begin{enumerate}
\item $\varphi(w,w)=0$;

\item \label{item:auau} $\varphi(a(u),a(u'))=0$ for every $u,u'\in B$ if and only if $u=u'$;

\item $\varphi(c_i,c_j)=1$ for every $c_i,c_j\in C$ if and only if $i=j$;

\item $\varphi(x,w)=0$ for every $x\in C\cup A(B)$;

\item \label{item:redverticesqualifya} $\varphi(a_i(v),w)=1$ for every $v\in R_i$ where $1\leq i\leq k$;

\item $\varphi(c_i,a(u))=1$ for every $c_i\in C$ and $a(u)\in A(B)$;

\item $\varphi(a(u),c_i)=0$ for every $a(u)\in A(B)$ and $c_i\in C$;

\item $\varphi(d_i,d_{i'})=0$ for every $d_i,d_{i'}\in A_{dummy}$;

\item $\varphi(d_i,w)=0$ for every $d_i\in A_{dummy}$;

\item \label{item:dummyindividual} $\varphi(d_i,x)=1$ for every $d_i\in A_{dummy}$ and every $x\in N\setminus (A_{dummy}\cup \{w\})$;

\item \label{item:individualdummy} $\varphi(x, d_i)=0$ for every $d_i\in A_{dummy}$ and every $x\in N\setminus (A_{dummy}\cup \{w\})$;

\item \label{item:encodinggraph} $\varphi(a_i(v),a(u))=0$ for every $v\in R_i$ where $1\leq i\leq k$ and every $a(u)\in A(B)$ if and only if $(v,u)\in E$;

\item \label{item:encodinglabels} $\varphi(a_i(v),c_j)=1$ for every $a_i(v)\in R_i$ and $c_j\in C$ if and only if $i=j$; and

\item $\varphi(x,y)$ which is not defined above can be set arbitrarily.
\end{enumerate}

Now we show the correctness of the reduction.

${(\Rightarrow:)}$ Let $W$ be a labeled red-blue dominating set of $G$. We shall show that ${\mathGCPIthreetwoins}$ is a {\yesins}.

Let $U\subseteq N$ be the set consisting of the individual $w$ and all individuals that correspond to $R\setminus W$. That is, $U=S\cup \{a_i(v)\mid v\in R_i\setminus W, 1\leq i\leq k\}$. Since $\varphi(w,w)=0$, and every individual corresponding to some vertex in $R$ qualifies $w$ (see~\ref{item:redverticesqualifya}), it holds that $w\in f^{(s,2)}(\varphi,U)$.

Now, let's consider the profile restricted to $N\setminus U$. Observe that
\[\left(N\setminus U\right) \cap \left(\bigcup_{1\leq i\leq k}A_i\right)=\left\{a_i(v)\mid v\in R_i\cap W, 1\leq i\leq k\right\}.\]
Let $a(u)$ be a candidate in $A(B)$ with $u\in B$. According to the construction of $\varphi$ and the fact that $W$ dominates $B$, there is at least one individual $a_i(v)$, corresponding to a vertex $v\in W$ dominating $u$, who disqualifies $a(u)$ (see~\ref{item:encodinggraph}). Since $\varphi(a(u),a(u))=0$ (see~\ref{item:auau}), it holds that $a(u)\not\in f^{(s,2)}(\varphi,N\setminus U)$. Since this holds for every $a(u)\in A(B)$, we have that $A(B)\cap f^{(s,2)}(\varphi,N\setminus U)=\emptyset$. On the other hand, for every $1\leq i\leq k$, since $|W\cap R_i|\leq 1$, $N\setminus U$ contains at most one individual $a_i(v)\in A_i$. According to the construction of $\varphi$, for every $c_i\in C$ only the following $s-1$ individuals in $N\setminus U$ qualify $c_i$:
\begin{enumerate}
\item[(1)] $c_i$ herself;

\item[(2)] $a_i(v)\in A_i$ where $v\in W$ (see~\ref{item:encodinglabels}); and

\item[(3)] all $s-3$ dummy individuals (see~\ref{item:dummyindividual}).
\end{enumerate}
It directly follows that $c_i\not\in f^{(s,2)}(\varphi,N\setminus U)$ for every $c_i\in C$. Finally, since $\varphi(d_i,d_{i'})=0$ for every $d_i,d_{i'}\in A_{dummy}$ and all individuals in $N\setminus U$ disqualify all dummy individuals, it holds that $d_i\not\in f^{(s,2)}(\varphi,N\setminus U)$ for every $d_i\in A_{dummy}$. In conclusion, $(A(B)\cup C\cup A_{dummy})\cap f^{(s,2)}(\varphi,N\setminus U)=\emptyset$. Now, it is easy to verify that $\varphi(x,w)=1$ for every $x\in (f^{(s,2)}(\varphi,U)\cup f^{(s,2)}(\varphi,N\setminus U)\setminus \{w\})$. As a result, $w\in f^{(s,2)}(\varphi, f^{(s,2)}(\varphi,U)\cup f^{(s,2)}(\varphi,N\setminus U))$.

$(\Leftarrow:)$ Let $U\subseteq N$ such that $w\in f^{(s,2)}(\varphi, f^{(s,2)}(\varphi,U)\cup f^{(s,2)}(\varphi,N\setminus U))$. Due to symmetry, assume that $w\in U$. Since $\varphi(w,w)=0$, all the other individuals who disqualify $w$ must be in $N\setminus U$. That is, $A(B)\cup C\cup A_{dummy}\subseteq N\setminus U$. Moreover, all individuals in $A(B)\cup C\cup A_{dummy}$ must be eliminated in the profile restricted to $N\setminus U$, i.e., $(A(B)\cup C\cup A_{dummy})\cap f^{(s,2)}(\varphi,N\setminus U)=\emptyset$. Let $a(u)$ be a vertex in $A(B)$ with $u\in B$. Since $\varphi(a(u),a(u))=0$, to eliminate $a(u)$, at least one individual who disqualifies $a(u)$ must be in $N\setminus U$. Due to the construction of the profile, all individuals in $N\setminus U$ who disqualify $a(u)$, except $a(u)$ herself, are in $\bigcup_{1\leq i\leq k}A_i$. Hence, at least one $a_i(v)\in A_i$ with $v\in R_i$ who disqualifies $a(u)$ must be in $N\setminus U$. According to the construction, the vertex $v$ dominates $u$ in the graph $G$. This implies that $W=\{v\in R\mid a_i(v)\in N\setminus U\}$ dominates $B$. Now, we show that $W$ contains at most one vertex in each $R_i$ where $1\leq i\leq k$. Let $c_i$ be an individual in $C$ where $1\leq i\leq k$. Since $\varphi(c_i,c_i)=1$, in order to eliminate each $c_i$, at most $s-1$ individuals who qualify $c_i$ can be in $N\setminus U$. According to the construction of the profile, all the $s-3$ dummy individuals in $A_{dummy}$ qualify $c_i$. Moreover, all individuals in $A_i$ qualify $c_i$. According to the above discussion, at most one of the individuals in $A_i$ can be in $N\setminus U$. Due to the definition of $\varphi$, this implies that $|W\cap R_i|\leq 1$. Now, it is easy to see that $W$ is a solution of the instance $\mathcal{I}$.
\end{proof}

\section{Conclusion}
\label{sec_conclusion}
We have studied the complexity of the group control by adding individuals ({\gcii}), group control by deleting individuals ({\gcei}), and group control by partitioning of individuals ({\gcpi}) problems for the consent rules $f^{(s, t)}$, the consensus-start-respecting rule $f^{CSR}$, and the liberal-start-respecting rule $f^{LSR}$, where in each problem an external agent has an incentive to make a given subset of individuals socially qualified by adding, or deleting a limited number of individuals, or by partitioning the set of individuals. In particular, as summarized in Table~\ref{tab:dichtomyforconsetrules},
we achieved dichotomy results for all three group control problems for consent rules, with respect to the values of the consent quotas. In addition, we studied the {\nph} problems from the parameterized complexity point of view, with respect to the size of $S$, the set of individuals whom the external agent wants to make socially qualified. We proved that {\gcii} and {\gcei} for consent rules are generally {\fpt}. On the other hand, {\gcpi} remains {\nph} for some consent rules even when $|S|=1$, excluding the possibility that the {\gcpi} problem for consent is {\fpt}, unless the parameter hierarchy collapses at some level.

Table~\ref{tab:dichtomyforconsetrules} shows that almost all social rules studied in this paper resist the three different control types, in the sense that either control problems for these rules are {\nph} or these rules are immune to the corresponding control types. Only the {\gcei} problem for consent rules $f^{(s,2)}$ and for $f^{CSR}$ are polynomial-time solvable. From the parameterized complexity point of view, {\gcii} and {\gcei} for consent rules are {\fpt} while the  procedural rule
$f^{LSR}$ is immune to {\gcei} and {\gcpi}.  So, we can conclude that the procedural rule $f^{LSR}$ outperforms the consent rules and $f^{CSR}$ in terms of resistance to control behavior. Note that whether {\gcii} for the two procedural rules is {\fpt} with respect to $|S|$ remains open. Moreover, whether {\gcpi} for $f^{CSR}$ is {\nph} remains open.

Following the workshop version of this paper, there have been other papers which look at similar problems. In particular, Erd\'{e}lyi, Reger, and Yang~\cite{AAMAS17ErdelyiRYBriberyControlGroupIdentification} extended our study to destructive control and constructive/destructive bribery problems in group identification. In addition, Erd\'{e}lyi, Reger, and Yang~\cite{DBLP:conf/aldt/ErdelyiRY17} also considered possibly and necessarily socially qualified individuals problems in group identification. Nevertheless, there still remain many directions for future research. For instance, for consent rules, we assume that the consent quotas do not change. It would be quite natural to consider group control problems where $s$ and $t$ depend on the number of individuals, say, e.g.,  $s$ is $30\%$ percent and $t$ is $35\%$ of the number of individuals. In addition, it is interesting to investigate the complexity of group control problems in restricted domains of dichotomous preferences~\cite{DBLP:conf/ijcai/ElkindL15}. One can also study control and bribery problems in generalized group identification~\cite{ChoJu} where the aim is to classify the individuals into multiple groups instead of putting them into only two classes.

\section*{Acknowledgement}
We thank Shao-Chin Sung and the anonymous reviewers of JAAMAS and COMSOC 2016 for their valuable comments.

\bibliographystyle{spmpsci}

\newpage\section*{Appendix}
The appendix is devoted to the proof of Lemma~\ref{lem:multicolorredbluedominatingisnphard}. In particular, we prove the {\nphns} of the {{LRBDS}} problem by a reduction from the {\sc{Red-Blue Dominating Set}} problem which is {\nph}~\cite{garey}.

\begin{description}
\item[{\sc{Red-Blue Dominating Set}}] (RBDS)
  \begin{itemize}
  \item[\textit{Input}:] A bipartite graph $G=(R\uplus B,E)$ and an integer $k$.
  \item[\textit{Question}:] Is there a subset
    $W\subseteq R$ such that $|W|\leq k$ and $W$ dominates $B$?
  \end{itemize}
\end{description}

Let $I'=(G'=(R'\uplus B', E'),k)$ be an instance of the RBDS problem. We construct an instance $I=(G=(R\uplus B, E), \{1,2,\dots,k\})$ for the LRBDS problem as follows.
For each vertex $u\in B'$, we create a vertex $\bar{u}\in B$. For each vertex $v\in R'$, we create $k$ vertices $v(1),\dots,v(k)\in R$, where the vertex $v(i)$ is labeled with $i$. Let $R_i$ be the set of the vertices in $R$ that have label $i$. The edges of the graph $G$ are defined as follows. If there is an edge $\edge{v}{u}\in E'$, then for every $1\leq i\leq k$ we create an edge between $v(i)$ and $\bar{u}$. This finishes the construction. It clearly takes polynomial time.

Suppose that $I'$ has a solution $W'$ of size $k'\leq k$. Let $(v_{x(1)},v_{x(2)},\dots,v_{x(k')})$ be any arbitrary order of the vertices in $W'$. Let $W=\{v(i)\mid v_{x(i)}\in W', 1\leq i\leq k'\}$. It is clear that no two vertices in $W$ have the same label, that is, $|W\cap R_i|\leq 1$ for every $1\leq i\leq k$. We shall show that $W$ dominates $B$. Let $u$ be a vertex in $B'$. Since $W'$ dominates $B'$, there is a vertex $v_{x(i)}\in W'$ such that $\edge{v_{x(i)}}{u}\in E'$. Then, according to the construction of $G$, we know that $\edge{v(i)}{\bar{u}}\in E$. Since this holds for every $u\in B'$, $W$ dominates $B$.

Suppose that $I$ has a solution $W$. We assume that for every $v\in R'$, $W$ contains at most one of $\{v(1),v(2),\dots,v(k)\}$. Indeed if $W$ contains two vertices $v(i)$ and $v(j)$ where $1\leq i\neq j\leq k$, then we could get a new solution $W\setminus \{v(j)\}$ for $I$, since according to the construction of the graph $G$, $v(i)$ and $v(j)$ have the same neighborhood, implying that a vertex in $B$ dominated by $v(j)$ is dominated by $v(i)$.
Let $W'=\{v\in R'\mid v(i)\in W, 1\leq i\leq k\}$. Let $u$ be a vertex in $B'$. Since $W$ is a solution of $I$, there is a vertex $v(i)\in W$ which dominates $\bar{u}\in B$. Then, according to the construction of the graph $G$, the vertex $v\in W'$ dominates $u$. It follows that $W'$ dominates $B'$.
\end{document}